\documentclass[12pt]{article}

\usepackage{amsmath,amssymb,amsthm}
\usepackage{tikz}
\usepackage{color}
\usepackage{comment}
\usepackage{enumerate}
\usepackage{complexity}
\usepackage{cite,url,hyperref}

\newtheorem{theorem}{Theorem}[section]
\newtheorem{lemma}[theorem]{Lemma}
\newtheorem{corollary}[theorem]{Corollary}
\newtheorem{proposition}[theorem]{Proposition}

\usepackage[normalem]{ulem}

\DeclareMathOperator {\gp} {gp}

\DeclareMathOperator{\ti}{ti} % tempo de iteração
\DeclareMathOperator{\I}{I} % interval function
\DeclareMathOperator{\conv}{conv}% interval function
\DeclareMathOperator {\gc} {g}
\DeclareMathOperator {\pc} {p3}
\DeclareMathOperator {\psc} {p3*}
\DeclareMathOperator {\mc} {m}

\DeclareMathOperator{\diss}{diss} %dissociation number
\DeclareMathOperator{\vc}{vc} %vertex cover
\DeclareMathOperator{\nd}{nd} %neighborhood diversity
\DeclareMathOperator{\edg}{edg}

\newcommand{\NN}{\mathbb{N}}

\usepackage{authblk}
\title{The iteration time and the general position number in graph convexities}
\author[1]{Julio Araujo}
\author[2]{Mitre C. Dourado}
\author[3]{Fábio Protti}
\author[1]{Rudini Sampaio\thanks{Email: \texttt{julio@mat.ufc.br, mitre@ic.ufrj.br, fabio@ic.uff.br, rudini@dc.ufc.br}}}
\affil[1]{Universidade Federal do Cear\'a, Fortaleza, Brazil}
\affil[2]{Inst. Computação, Universidade Federal do Rio de Janeiro, Brazil}
\affil[3]{Inst. Computação, Universidade Federal Fluminense, Niterói, Brazil}

\begin{document}
\maketitle
\begin{abstract}
In this paper, we study two graph convexity parameters: iteration time and general position number. The iteration time was defined in 1981 in the geodesic convexity, but its computational complexity was so far open. The general position number was defined in the geodesic convexity and proved $\NP$-hard in 2018. We extend these parameters to any graph convexity and prove that the iteration number is $\NP$-hard in the $P_3$ convexity. We use this result to prove that the iteration time is also $\NP$-hard in the geodesic convexity even in graphs with diameter two, a long standing open question. These results are also important since they are the last two missing $\NP$-hardness results regarding the ten most studied graph convexity parameters in the geodesic and $P_3$ convexities. We also prove that the general position number of the monophonic convexity is $W[1]$-hard (parameterized by the size of the solution) and $n^{1-\varepsilon}$-inapproximable in polynomial time for any $\varepsilon>0$ unless $\P=\NP$, even in graphs with diameter two. Finally, we also obtain FPT results on the general position number in the $P_3$ convexity and we prove that it is $W[1]$-hard (parameterized by the size of the solution).
\end{abstract}


%-------------------------------------------------------
%-------------------------------------------------------
%-------------------------------------------------------
\section{Introduction}\label{sec:intro}

Convexity is a classical topic, studied in many different branches of mathematics.
A rich source is the book ``Theory of Convex Structures'' by van de Vel~\cite{vandevel93}.
The study of convexities applied to graphs has started recently, about 50 years ago.
%The 1972 paper of Erd{\H o}s et al. \cite{erdos72} is one of the first in this topic, focused in tournments.
Accordingly to Duchet \cite{duchet87}, the first paper on general graphs, published in english, is the 1981 paper ``Convexity in graphs'' from Frank Harary and Juhani Nieminem \cite{harary81}, where it was introduced the iteration time of a graph in the geodesic convexity.
Other important graph convexity parameter is the general position number, which is related to the 1917 Dudeney's No-three-in-line problem \cite{dudeney-1917}. The explicit definition of it to the graph geodesic convexity was done by Manuel and Klavžar \cite{manuel-2018a}.

To the best of our knowledge, these parameters were defined and studied only in the geodesic convexity.
Our first contribution in this paper is the definition of them to any graph convexity. We also prove computational complexity results on these parameters in three important graph convexities (geodesic, monophonic and $P_3$), including the NP-hardness of the geodesic iteration number, a long standing open question. In order to prove this, we also had to prove the NP-hardness of the $P_3$ iteration number.
For this, we need some terminology.

A \emph{convexity} $\mathcal{C}$ \cite{vandevel93} on a finite set $V\ne\emptyset$ is a family of subsets of $V$ such that $\emptyset,V\in \mathcal{C}$ and $\mathcal{C}$ is closed under intersections. That is, $S_1,S_2\in\mathcal{C}$ implies $S_1\cap S_2\in\mathcal{C}$. A member of $\mathcal{C}$ is said to be a $\mathcal{C}$-\emph{convex set}. Given $S\subseteq V$, the $\mathcal{C}$-\emph{convex hull} of $S$ is the smallest $\mathcal{C}$-convex set $\conv_\mathcal{C}(S)$ containing $S$. %We say that $S$ is a $\mathcal{C}$-\emph{hull set} if $\conv_\mathcal{C}(S) = V$.
It is easy to see that $\conv_\mathcal{C}(\cdot)$ is a \emph{closure operator}, that is, for every $S,S'\subseteq V$:
\begin{enumerate}[(a)]
    \item $S\subseteq\conv_\mathcal{C}(S)$ (extensivity),
    \item $S\subseteq S' \Rightarrow \conv_\mathcal{C}(S)\subseteq\conv_\mathcal{C}(S')$ (monotonicity),
    \item $\conv_\mathcal{C}(\emptyset) = \emptyset$ (normalization\footnote{Here we follow the definition of closure operator from van de Vel \cite{vandevel93}, which includes the normalization property.}) and
    \item $\conv_\mathcal{C}(\conv_\mathcal{C}(S)) = \conv_\mathcal{C}(S)$ (idempotence).
\end{enumerate}

We say that $\I:2^V\to 2^V$ is an \emph{interval function} on $V$ if, for every $S,S'\subseteq V$,
(a) $S\subseteq\I(S)$ [extensivity],
(b) $S\subseteq S'\ \Rightarrow\ \I(S)\subseteq\I(S')$ [monotonicity] and
(c) $\I(\emptyset)\ =\ \emptyset$ [normalization].
It is not difficult to prove that every interval function induces a unique convexity, containing each set $S\subseteq V$ such that $\I(S)=S$. %, usually called \emph{$I$-closed}, or \emph{convex}. 
Moreover, every convexity is induced by an interval function.
In the rest of the paper, we assume that every convexity $\mathcal{C}$ on $V$ is defined by an explicitly given interval function $\I_\mathcal{C}(\cdot)$ on $V$.
%
%
%We say that $\mathcal{C}$ is an \emph{interval convexity} on $V$ if it is induced by an interval function $\I_\mathcal{C}:2^V\to 2^V$.
%In interval convexities,
It is well known that the convex hull of a set $S$ in a convexity $\mathcal{C}$ can be obtained by exhaustively applying the corresponding interval function $\I_\mathcal{C}(\cdot)$ until obtaining a convex set.

Given a finite graph $G$, a \emph{graph convexity} on $G$ is simply a convexity $\mathcal{C}$ on $V(G)$ with a given interval function $\I_\mathcal{C}(\cdot)$ on $V(G)$. A standard way to define a \emph{graph convexity} $\mathcal{C}$ on a graph $G$ is by fixing a family $\mathcal{P}$ of paths of $G$ and 
taking the interval function $\I_\mathcal{C}(S)$ as the set with all vertices lying on some path of $\mathcal{P}$ whose endpoints are in $S$.
The most studied graph convexities are path convexities, such as the \emph{geodesic convexity}~\cite{everett85,faja,harary81}, the \emph{monophonic convexity}~\cite{duchet88,jamison82}, the $P_3$ \emph{convexity}~\cite{centeno09} and the $P_3^*$ \emph{convexity} \cite{araujo13-lagos}, where $\mathcal{P}$ is, respectively, the family of all geodesics (shortest paths) of the graph, of all induced paths, of all paths of order three and of all induced paths of order three.

We use the subscripts $\gc$, $\mc$, $\pc$  and $\psc$ to refer to the the geodesic convexity, the monophonic convexity, the $P_3$ convexity and the $P_3^*$ convexity, respectively.
In the following, when we say that some numerical parameter is \NP-hard, we mean that the determination of its value is an \NP-hard problem.

%When there is no risk of confusion, sometimes we omit $\mathcal{C}$ and write only \emph{convex set}, \emph{convex hull}, $\conv(S)$ and $\I(S)$.

%The interval function of the geodesic convexity associates to every subset $S\subseteq V(G)$ the set $\I_{\gc}(S)$, which is the set $S$ and every vertex in a shortest path between two vertices of $S$. The interval function of the $P_3$ convexity associates to every subset $S\subseteq V(G)$ the set $\I_{\p3}(S)$, which is the set $S$ and every vertex with at least two neighbors in $S$. We say that $S\subseteq V(G)$ is convex in the geodesic convexity (resp. $P_3$ convexity) if $\I_{\gc}(S)=S$ (resp. $\I_{\p3}(S)=S$). The \emph{convex hull} $\conv(S)$ of $S\subseteq V(G)$ in a graph convexity is the minimum convex set which contains $S$.
From these definitions, we can define the two parameters investigated in this paper for any graph convexity.
The \emph{iteration time of a set} $S\subseteq V(G)$ in a graph convexity $\mathcal{C}$, denoted by $\ti_\mathcal{C}(S)$, is the minimum $k$ such that $\I^k_\mathcal{C}(S)=\conv_\mathcal{C}(S)$, that is, $k$ applications of the interval function are necessary in order to obtain the convex hull of $S$. 
Also let the \emph{iteration time of a graph} $G$ in a graph convexity $\mathcal{C}$, denoted by $\ti_\mathcal{C}(G)$, be the maximum value of $\ti_\mathcal{C}(S)$ among the subsets $S$ of $V(G)$.
The iteration time is one of the first graph convexity parameters, introduced in 1981 by Harary and Nieminem \cite{harary81} in the geodesic convexity.
%There are few results in the literature regarding it and only recently has this parameter been studied again.
In \cite{parvathy99}, bounds were obtained for the geodesic iteration time.
In 2016, Dourado et al. \cite{tempo-it16} obtained a polynomial time algorithm to determine the geodesic iteration time on distance hereditary graphs, which was improved in 2020 by Moscarini \cite{moscarini20}.
However, the computational complexity of the iteration time was so far open for any graph convexity, including the geodesic convexity.

The \emph{general position number} of a graph $G$ in a graph convexity $\mathcal{C}$, denoted by $\gp_{\mathcal{C}}(G)$, is the size of a maximum subset of $V(G)$ in general position, where we say that a subset $S\subseteq V(G)$ is in \emph{general position} if $z\not\in\I_\mathcal{C}(\{x,y\})$ for every distinct $x,y,z\in S$.
In the $P_3$ convexity, curiously any subset in general position induces a subgraph with maximum degree 1 and then $\gp_{\pc}(G)$ is equivalent to the \emph{dissociation number} $\diss(G)$ of the graph $G$, a parameter introduced in 1981 by Yannakakis \cite{yannakakis81}, who also proved $\NP$-hardness even for bipartite graphs and planar graphs with maximum degree four.
In the $P_3^*$ convexity, any subset in general position induces a subgraph whose connected components are cliques and, then, $\gp_{\psc}(G)$ is equivalent to the \emph{IUC number} (\emph{independent union of cliques}) of the graph $G$, a parameter introduced in 2020 by Ertem et al. \cite{IUC20} and proved $\NP$-hard even in planar graphs.

The geodesic general position number is a generalization of the No-three-in-line problem in the $n\times n$ grid from discrete geometry, which can be traced to the famous Dudeney's ``Puzzle with Pawns'' of his book ``Amusements in Mathematics'' \cite{dudeney-1917} from 1917.
All the following results are for the geodesic convexity.
In 1995, Korner~\cite{korner-1995} investigated the general position number on hypercubes, while in~\cite{ullas-2016} it was considered for the first time on general graphs. However, the formalization of the problem as we know it today and the notation that is in use have been introduced in~\cite{manuel-2018a, manuel-2018b}. Also see~\cite{froese-2017, ku-2018, misiak-2016, payne13} for the related general position subset selection problem in computational geometry.
In 2018, it was proved that determining the general position number is \NP-hard \cite{manuel-2018a}.
In 2019, general position sets in graphs were characterized \cite{bijo-2019} and, after this, several additional papers on the general position problem were published, many of them with bounds on the maximum size of a general position set and exact values in graph products \cite{thomas-2020, klavzar-2021d, klavzar-2021a, klavzar-2021b, klavzar-2019, patkos-2020,  tian-2021a,tian-2021b}.

In this paper, we prove that the iteration time is \NP-hard in the $P_3$ convexity. This result is the basis to prove that the iteration time is also \NP-hard in the geodesic convexity even in graphs with diameter two. Surprisingly, even though this problem has been defined in 1981, its computational complexity had not been settled yet.

In addition to these points, we consider these results even more important due to the fact that they are the last two missing proofs of \NP-hardness among the ten main graph convexity parameters in the geodesic and $P_3$ convexities.
In order to justify this, notice that the following nine graph convexity parameters are \NP-hard in the geodesic and $P_3$ convexities: hull number and convexity number \cite{Araujoetal2013,araujo13-lagos,Douradoetal2012}, interval number \cite{centeno09,harary93}, Carathéodory number \cite{barbosa-2012,erika15}, Radon number \cite{araujo18,radon13}, Helly number \cite{Carvalho2016,mitre-aline17}, general position number \cite{manuel-2018a,yannakakis81}, rank \cite{rank17,rank14} and percolation time \cite{araujo18,marcilon18}.

Regarding the monophonic convexity, we prove that the monophonic general position number is $W[1]$-hard (parameterized by the size of the solution) and $n^{1-\varepsilon}$-inapproximable in polynomial time for any $\varepsilon>0$ unless $\P=\NP$, even in graphs with diameter two.
Finally, we also prove that the general position number in the $P_3$ convexity is $W[1]$-hard (parameterized by the size $k$ of the solution) and we obtain fixed parameter tractable results regarding the neighborhood diversity, the vertex convex number and the cliquewidth of the graph. 

%-------------------------------------------------------
%-------------------------------------------------------
%-------------------------------------------------------
\section{Preliminary results}

Given a positive integer $n$, let $[n]=\{1,\ldots,n\}$. 
Given a graph $G$ and $S\subseteq V(G)$, we say that a vertex $v$ has \emph{iteration time $k>0$ starting from} $S$ if $v\in I^k_{\mathcal{C}}(S)\setminus I^{k-1}_{\mathcal{C}}(S)$, and that the vertices of $S$ have iteration time 0 starting from $S$.

\begin{lemma}
Let $n\geq 4$ be an integer.
In the complete graph $K_n$ and the cycle $C_n$, we have:
\begin{itemize}
\item $\ti_{\pc}(K_n)=1$ and $\ti_{\gc}(K_n)=\ti_{\mc}(K_n)=0$,
\item $\ti_{\pc}(C_n)=\ti_{\gc}(C_n)=\ti_{\mc}(C_n)=1$.
%\item $\ti_{\pc}(P_n)=\ti_{\gc}(P_n)=\ti_{\mc}(P_n)=1$, and
\end{itemize}
\end{lemma}

\begin{proof}
Let $S$ be a proper subset of at least two vertices of $K_n$.
In the $P_3$ convexity, $\I_{\pc}(S)$ contains all vertices of $K_n$, since every vertex in $V(K_n)\setminus S$ has two neighbors in $S$. Thus $\ti_{\pc}(K_n)=1$.
In the monophonic convexity, $\I_{\mc}(S)=\conv_{\mc}(S)=S$, since no vertex is an induced path between two vertices of $K_n$. Thus $\ti_{\mc}(K_n)=0$. Consequently, in the geodesic convexity, $\ti_{\gc}(K_n)=0$.

Now let $S$ be a proper subset of at least two non-adjacent vertices of $C_n$.
In the $P_3$ convexity, $v\in\I_{\pc}(S)$ if $v\in S$ or its two neighbors are in $S$. Then $\I_{\pc}(S)=\conv_{\pc}(S)$. Therefore $\ti_{\pc}(C_n)=1$.
In the monophonic convexity, $\I_{\mc}(S)$ contains all vertices of $C_n$, since every vertex is in an induced path between two non-adjacent vertices of $C_n$. Thus, $\ti_{\mc}(C_n)=1$.
Moreover, in the geodesic convexity, $\I_{\gc}(S)$ contains all vertices if and only if $S$ contains three vertices $x,y,z$ such that the distance between $x$ and $y$ is smaller than them sum of the distances between $x$ and $z$, and between $z$ and $y$. In this case, $\ti_{\gc}(S)=\ti_{\gc}(\{x,y,z\})=1$. Otherwise, $\I_{\gc}(S)=\I_{\gc}(\{x,y\})$, where $x$ and $y$ are the vertices of $S$ with maximum distance.
Therefore $\ti_{\gc}(C_n)=1$.

\end{proof}

\begin{lemma}
Let $T$ be a tree with at least three vertices.
Then $\ti_{\gc}(T)=\ti_{\mc}(T)=1$.
Moreover $\ti_{\pc}(T)=k\geq 1$ if and only if $T$ has a path $v_1,v_2,\ldots,v_k$ such that $v_i$ has degree at least 3 for every $1\leq i<k$ and $v_k$ has degree at least two.
\end{lemma}

\begin{proof}
In a tree, $\I_{\mc}(S)=\I_{\gc}(S)=\conv_{\mc}(S)=\conv_{\gc}(S)$ for every $S\subseteq V(T)$. Thus, since $T$ has two non-adjacent vertices, $\ti_{\gc}(T)=\ti_{\mc}(T)=1$.

Now consider the $P_3$ convexity on $T$.
Let $S\subseteq V(T)$ with $\ti(S)=k\geq 1$.
Then $S$ contains a path $v_1,v_2,\ldots,v_k$ such that $v_i$ has iteration time $i$, starting from $S$, for every $i\in[k]$.
Clearly, for any $1<i<k$, $v_i$ must have at least one neighbor other than $v_{i-1}$ and $v_{i+1}$, and consequently the degree of $v_i$ is at least 3.
Moreover, $v_1$ must have at least two neighbors other than $v_2$, and $v_n$ must have at least one neighbor other than $v_{k-1}$.
Therefore, $v_i$ has degree at least 3 for every $1\leq i<k$ and $v_k$ has degree at least two.

Finally suppose that $T$ has a path $v_1,v_2,\ldots,v_k$ such that $v_i$ has degree at least 3 for every $1\leq i<k$ and $v_k$ has degree at least two. Let $S\subseteq V(T)$ be such that $S\cap\{v_1,\ldots,v_k\}=\emptyset$ and $S$ contains exactly two neighbors of $v_1$ and exactly one neighbor of $v_i$ for every $1<i\leq k$.
Notice that $\ti_{\pc}(S)=k$, since $v_i$ has iteration time $i$ for every $i\in[k]$. Thus, $\ti_{\pc}(T)\geq k$ and we are done.

\end{proof}

As mentioned before, the general position number was investigated in the geodesic and $P_3$ convexities in many papers. The lemma below is about this parameter in the  monophonic convexity in simple graphs, such as the \emph{wheel graphs} $W_n$ for $n\geq 4$, obtained from the cycle $C_{n-1}$ by adding a universal vertex.

\begin{lemma}
Let $n\geq 4$ be an integer.
In the monophonic convexity, $\gp_{\mc}(C_n)=\gp_{\mc}(P_n)=2$ and $\gp_{\mc}(K_n)=n$. Moreover, $\gp_{\mc}(W_4)=4$ and $\gp_{\mc}(W_n)=3$ for $n\geq 5$.
\end{lemma}

\begin{proof}
First consider the monophonic convexity.
In the complete graph $K_n$, no vertex is in an induced path between two vertices $x$ and $y$, other than $x$ and $y$. Then $V(K_n)$ is in general position. Thus $\gp_{\mc}(K_n)=n$.
In the cycle $C_n$, every vertex is an induced path between two non-adjacent vertices. Then $\gp_{\mc}(C_n)=2$.
In the path $P_n$, $\conv_{\mc}(S)=\I_{\mc}(S)=\I_{\mc}(\{x,y\})$, where $x$ and $y$ are the vertices of $S$ with maximum distance in $P_n$.
Then $\gp_{\mc}(P_n)=2$.

In wheel graphs, we have that $W_4=K_4$ and then $\gp_{\mc}(W_4)=4$. For $n\geq 5$, every vertex is in an induced path between two non-adjacent vertices of the main cycle of $W_n$. With this, a maximum general position set must have the universal vertex and two adjacent vertices of the main cycle and, then, $\gp_{\mc}(W_n)=3$ for $n\geq 5$. 
\end{proof}

%-------------------------------------------------------
%-------------------------------------------------------
%-------------------------------------------------------
\section{Iteration time is \NP-hard in the P3 convexity}

We first prove a lemma that will be useful in the main theorem of this section.

\begin{lemma}\label{lema2}
Let $G$ be a graph and $k$ be a positive integer.
If $S\subseteq V(G)$ and $\ti_{\pc}(S)\geq k$, then the subgraph $H$ of $G$ induced by $\conv_{\pc}(S)$ contains a path $v_1,\ldots,v_k$ of vertices not in $S$, where $v_1$ is adjacent to two vertices of $S$, the degree of $v_i$ in $H$ is at least 3 for every $1\leq i<k$ and the degree of $v_k$ in $H$ is at least 2.
\end{lemma}

\begin{proof}
Let $H$ be the subgraph of $G$ induced by $\conv_{\pc}(S)$.
Since $\ti_{\pc}(S)\geq k$, then, starting from $S$, the vertices of $S$ have iteration time 0 and, for every $i\in[k]$, there is a vertex $v_i$ with iteration time $i$ adjacent to a vertex with iteration time $i-1$. Therefore, $H$ has a path $v_1,v_2,\ldots,v_k$ of vertices not in $S$ such that $v_i$ has iteration time $i$ for every $i\in[k]$, $v_1$ is adjacent to two vertices of $S$ and $v_i$ is adjacent to a vertex outside this path whose iteration time is at most $i-1$, and we are done.
\end{proof}

The following theorem contains the main result of this section, which is the basis of the proof of the main theorem of the next section.

\begin{theorem}\label{teo1}
Given a graph $G$ and a positive integer $k$, deciding whether the $P_3$ iteration time $\ti_{\pc}(G)$ is at least $k$ is an \NP-complete problem even in bipartite graphs.
\end{theorem}

\begin{proof}
Let us prove that the $P_3$ iteration time problem is \NP-complete by showing a polynomial reduction from the problem \textbf{3-SAT}.
Given $m\geq 10$ clauses $\mathcal{C}=\{c_1,\ldots,c_m\}$ on variables $X=\{x_1, \ldots,x_n\}$ of an instance of \textbf{3-SAT}, let us denote the three literals of $c_i$ by $\ell_{i,1}$, $\ell_{i,2}$ and $\ell_{i,3}$.
Let $k=2m$ and let $G$ constructed as follows.

For each clause $c_i$ of $\mathcal{C}$, add to $G$ the vertices $c_i$, $c'_i$, $p'_i$, $\ell_{i,1}$, $\ell_{i,2}$ and $\ell_{i,3}$, and add the edges $p'_ic'_i$, $c'_ic_i$, $c_i\ell_{i,1}$, $c_i\ell_{i,2}$ and $c_i\ell_{i,3}$. Create a vertex $p''_1$ and add the edge $p''_1c'_1$.
Also add the edge $c_ic'_{i+1}$ for every $1\leq i<m$.
Moreover, for each pair of literals $\ell_{i,a}$ and $\ell_{j,b}$ such that one is the negation of the other, add the vertices $w_{i,j,a,b}$ and $w'_{i,j,a,b}$ adjacent to the vertices $\ell_{i,a}$ and $\ell_{j,b}$.
Also add the vertices $w_i$ and $w'_i$, and add the edges $w_i\ell_{i,p}$ and $w'_i\ell_{i,p}$ for every $i\in[m]$ and $p\in[3]$.
Let $W$ be the set of all vertices $w_{i,j,a,b}$, $w'_{i,j,a,b}$, $w_i$ and $w'_i$.
Finally, add vertices $z$ and $z'$ adjacent to all vertices in $W$.
Let $L$ be the set of all vertices $\ell_{i,a}$ and let $C$ be the set of all vertices $c_i$.
Notice that $G$ is bipartite.

We will prove that $\mathcal{C}$ is satisfiable if and only if $G$ contains a set $S\subseteq V(G)$ such that $\ti_{\pc}(G)\geq k=2m$.

Suppose that $\mathcal{C}$ has a truth assignment. For each clause $c_i$, let $a_i\in\{1,2,3\}$ such that $\ell_{i,a_i}$ is true for all $i\in[m]$.
Let $S = \{\ell_{i,a_i},\ p'_i : i\in[m]\}\cup\{p''_1\}$.
Notice that no vertex of $W$ is in $\conv(S)$, since $S$ was obtained from a truth assignment.
Moreover, starting from $S$, the iteration time of $c_i$ is $2i$ for every $i$.
Consequently, the iteration time of $G$ is at least $2m$.

Now, suppose that $G$ has a set $S$ with iteration time $\ti_{\pc}(S)=2m\geq 20$. Let $H$ be the subgraph induced by $\conv_{\pc}(S)$.
If $H$ contains two vertices $\ell_{i,a}$ and $\ell_{j,b}$ with iteration times at most 4 such that the literal $\ell_{i,a}$ is the negation of the literal $\ell_{j,b}$ (and vice-versa), then the iteration time of $w_{i,j,a,b}$ and  $w'_{i,j,a,b}$ is at most 5, the iteration time of $z$ and $z'$ is at most 6, the iteration time of the remaining vertices of $W$ is at most 7, the iteration time of the remaining vertices of $L$ is at most 8 and the iteration time of the vertices of $C$ is at most 9, a contradiction, since $2m\geq 20$.
Analogously, $H$ contains at most one vertex from $\ell_{i,1}$, $\ell_{i,2}$ and $\ell_{i,3}$, for every $i\in[m]$.
From the same argument, $H$ does not contain two vertices of $W$ with iteration times at most 5, nor $z$ and $z'$ with iteration times at most $6$.

With this, we conclude that $H$ contains at most one vertex of $W$, at most one vertex of $\{z,z'\}$ and at most one vertex of $\{\ell_{i,1},\ell_{i,2},\ell_{i,3}\}$ for every $i\in[m]$.
Therefore, if exactly one vertex of $\{z,z'\}$ belongs to $H$, its degree is at most 1 in $H$. If exactly one vertex of $W$ belongs to $H$, its degree is at most 2 in $H$.
Moreover, the degree of every $\ell_{i,a}$ in $H$ is at most 1, except at most one vertex with degree at most 2 in $H$. 

Then, from Lemma \ref{lema2}, the only possibility for iteration time $2m$ is the path $c'_1,c_1,c'_2,c_2,\ldots,c'_m,c_m$ of $G$ with iteration times $1,2,\ldots,2m$, respectively.
Since the iteration time of $c'_1$ is 1, we may assume that $p'_1,p''_1\in S$.
Then, for every vertex $c_i$, there must be at least one neighbor $\ell_{i,a}\in S$. Therefore, by assigning true to the literal $\ell_{i,a}$ for every vertex $\ell_{i,a}\in S$, we obtain a truth assignment, and we are done.
\end{proof}

%-------------------------------------------------------
%-------------------------------------------------------
%-------------------------------------------------------
\section{Iteration time is \NP-hard in the geodesic convexity}

The following lemma shows an important reduction from the $P_3$ convexity to the geodesic convexity, regarding the iteration time and the general position number.

\begin{lemma}\label{lema1}
Let $G$ be a graph and let $G_u$ be obtained from $G$ by adding a universal vertex $u$. 
If $G$ is a triangle free graph with at least 3 vertices, then $\ti_{\gc}(G_u)=\max\{\ti_{\pc}(G),1\}$ and $\gp_{\gc}(G_u)=\max\{\gp_{\pc}(G),\ \omega(G)+1\}$.
\end{lemma}

\begin{proof}
Since $G_u$ has the universal vertex $u$, then the diameter of $G_u$ is two. 
Note that every shortest path in $G_u$ between two non-adjacent vertices of $G$ is a $P_3$, which is induced because $G$ is triangle free.

If $G$ has no $P_3$, then $\ti_{\pc}(G)=0$ and $\ti_{\gc}(G_u)=1$, where $u$ is the vertex of $G_u$ with iteration time 1 in the geodesic convexity.
So assume that $G$ has an induced $P_3$ and consequently $\ti_{\pc}(G)\geq 1$ and $\ti_{\gc}(G_u)\geq 1$. We prove that $\ti_{\gc}(G_u)=\ti_{\pc}(G)$.
Let $S\subseteq V(G)$ with iteration time $t\geq 1$ in the $P_3$ convexity.
Clearly $S$ contains two non-adjacent vertices and then $u$ has iteration time 1 in the geodesic convexity on $G_u$ starting from $S$.
Since $u$ is universal, no vertex of $G$ is in a shortest path between $u$ and other vertex. That is, $u$ does not help to generate other vertices.
With this, we conclude that the iteration time of every vertex of $G$ in the $P_3$ convexity on $G$ starting from $S$ is equal to the iteration time of the same vertex of $G$ in the geodesic convexity on $G_u$ starting from $S$.
Thus the iteration time of $S$ in the geodesic convexity on $G_u$ is equal to the iteration time of $S$ in the $P_3$ convexity on $G$.
From the other hand, as mentioned above, the iteration time of $S$ in the $P_3$ convexity on $G$ is equal to the iteration time of $S\cup\{u\}$ in the geodesic convexity on $G_u$, and we are done.

Now let us deal with the general position number.
Let $S$ be a $P_3$ general position set of $G$.
Clearly $S$ is also a geodesic general position set of $G_u$.
If $S$ contains two non-adjacent vertices, then $S\cup\{u\}$ is not a geodesic general position set of $G_u$; otherwise, $S\cup\{u\}$ is a clique and a geodesic general position set of $G_u$.
From the other hand, let $S'$ be a geodesic general position set of $G_u$. If $S'$ contains $u$, then $S'$ must be a clique and then $S'\setminus\{u\}$ is a clique of $G$. Otherwise, $S$ is also a $P_3$ general position set of $G$, and we are done. 

\end{proof}

We now prove the long-standing open question of \NP-hardness of the geodesic iteration time. Notice that the proof strongly depends on the \NP-hardness of the $P_3$ iteration time (Theorem \ref{teo1}).

\begin{theorem}
Given a graph $G$ and an integer $k$, deciding whether the geodesic iteration time $\ti_{\gc}(G)$ is at least $k$ is an \NP-complete problem even in graphs with diameter two.
\end{theorem}

\begin{proof}
Let $G$ be a graph.
Consider the graph $G_u$ obtained from the addition of a universal vertex $u$ in $G$.
From the reduction of Lemma \ref{lema1} from the iteration time in the $P_3$ convexity, we have that $\ti_{\gc}(G_u)=\max\{\ti_{\pc}(G),1\}$.
If the diameter of $G$ is 0, then $\ti_{\pc}(G)=0$.
If the diameter of $G$ is 1 and each vertex has degree at most 1, then $\ti_{\pc}(G)=0$.
If the diameter of $G$ is 1 and there is a vertex with degree at least two, then $\ti_{\pc}(G)=1$.
Since the iteration time in the $P_3$ convexity is \NP-hard from Theorem \ref{teo1}, we may assume that $G$ has diameter at least two and then $\ti_{\pc}(G)\geq 1$.
Therefore $\ti_{\gc}(G_u)=\ti_{\pc}(G)$ and consequently the iteration time in the geodesic convexity is also \NP-hard.
\end{proof}

%-------------------------------------------------------
%-------------------------------------------------------
%-------------------------------------------------------
\section{General position number is NP-hard in the monophonic convexity}

As mentioned previously, there are many recent papers regarding the general position number in the geodesic convexity. Moreover, the general position number in the $P_3$ convexity is equivalent to the dissociation number, which has also been extensively studied in
the literature. In this section, we obtain the first complexity results on the general position number in other well investigated graph convexity, the monophonic convexity, which was introduced by Jamison \cite{jamison82} in 1982. See also the papers \cite{faja,duchet88} from Farber and Jamison in 1986 and Duchet in 1988 with some of the first results on the monophonic convexity.

We first prove in Theorem \ref{teo-mono1} that deciding if a set with 3 vertices is in general position in the monophonic convexity is a $\coNP$-complete problem. Later we prove in Theorems \ref{teo-mono2} and \ref{teo-mono3} that deciding if there is a set with $k$ vertices in general position in the monophonic convexity is W[1]-hard when parameterized by the size $k$ of the solution and that the general position number $\gp_{\mc}(G)$ of the monophonic convexity is highly inapproximable.

\begin{theorem}\label{teo-mono1}
Given a graph $G$ and a set $S\subseteq V(G)$, determining whether $S$ is in general position in the monophonic convexity is a $\coNP$-complete problem, even if $|S|=3$.  
\end{theorem}

\begin{proof}
A certificate that $S$ is not in general position consists of three distinct vertices $x,y,z\in S$ and an induced path $P$ from $x$ to $y$ passing through $z$. The hardness proof is a reduction from the following problem: given a graph $H$ and three specified vertices $x,y,z\in V(H)$, decide whether there is an induced path from $x$ to $y$ in $H$ passing through $z$. This problem is proved to be $\NP$-complete in \cite{Haas2006} (Theorem 10). Let $G = H + (E_x\cup E_y)$, where $E_x$ and $E_y$ are sets of edges defined as follows: $E_x=\{ab \mid ab\notin E(H) \ \text{and} \ a,b\in N_H(x)\}$ and  $E_y=\{ab \mid ab\notin E(H) \ \text{and} \ a,b\in N_H(y)\}$. In other words, $x$ and $y$ are simplicial vertices of $G$. In addition, define $S=\{x,y,z\}$. We prove that there is an induced path from $x$ to $y$ in $H$ passing through $z$ if and only if $S$ is not in general position in $G$ with respect to the monophonic convexity.

Suppose first that $P=x,v_1,\ldots,v_k,y$ is an induced path in $H$ passing through $z$, with $z=v_j$ for some $j\in\{1,\ldots,k\}$. Note that $v_1$ is the only neighbor of $x$ in $V(P)$. Similarly, $v_k$ is the only neighbor of $y$ in $V(P)$. Thus, $E(P)\cap E_x=\emptyset$ and $E(P)\cap E_y=\emptyset$, and this implies that $P$ is also an induced path in $G$. But this means that $S$ is not in general position in $G$ with respect to the monophonic convexity. 

Conversely, suppose that $S$ is not in general position in $G$ with respect to the monophonic convexity. Thus, there is an induced path $P$ in $G$ between two distinct vertices $p,q$ of $S$ such that the third vertex $r$ of $S$, $r\notin\{p,q\}$, is an internal vertex of $P$. Since $x$ is a simplicial vertex of $G$, we have that $r\neq x$. Likewise, $r\neq y$. Thus, $r=z$ and $P$ is an induced path from $x$ to $y$ in $G$ passing through $z$. In addition, observe that $E(P)\subseteq E(H)$. Thus, $P$ is an induced path in $H$, and the theorem follows.
\end{proof}

\begin{theorem}\label{teo-mono2}
Given a graph $G$ and an integer $k>0$, deciding whether the general position number $\gp_{\mc}(G)$ in the monophonic convexity is at least $k$ is $\NP$-hard, even in graphs with diameter two.
\end{theorem}

\begin{proof}
We obtain a polynomial reduction from the Clique problem, which has as an instance a graph $H$ and a positive integer $\ell$ and asks whether $H$ has a clique with size at least $\ell$.
Let $v_1,\ldots,v_n$ be the vertices of $H$ in which we may assume that $n\geq 3$ and $H$ does not have isolated vertices.
We build a graph $G$ from $H$, by adding for each vertex $v_i$ a new vertex $u_i$ adjacent to every vertex of $H$ except $v_i$.
Moreover, we include the new vertex $u$ adjacent to every vertex $u_1,\ldots,u_n$.
Notice that every induced path of $H$ is also an induced path in $G$.
Let $i,j,k\in[n]$ distinct.

First notice that $\I_{\mc}(\{v_i,u_i\})=V(G)$, since the vertices $u$ and $u_j$ belong to the induced path $v_i-u_j-u-u_i$ and, if $v_j$ is adjacent to $v_i$, then $v_j$ belongs to the induced path $v_i-v_j-u_i$, otherwise, $v_j$ belongs to the induced path $v_i-u_k-v_j-u_i$.

Now suppose that $v_i$ and $v_j$ are non-adjacent.
Thus $\I_{\mc}(\{u_i,u_j\})=V(G)$, since the vertex $u$ belongs to the induced path $u_i-u-u_j$, the vertex $v_k$ belongs to the induced path $u_i-v_k-u_j$ and the vertices $v_i$, $v_j$ and $u_k$ belong to the induced path $u_i-v_j-u_k-v_i-u_j$.
Let us prove that $\I_{\mc}(\{v_i,v_j\})=V(G)$.
If $v_k$ is adjacent to $v_i$ and $v_j$, then $v_k$ belongs to the induced path $v_i-v_k-v_j$.
If $v_k$ is adjacent to $v_i$, but not to $v_j$, then $v_k$ belongs to the induced path $v_i-v_k-u_i-v_j$.
If $v_k$ is adjacent to $v_j$, but not to $v_i$, then $v_k$ belongs to the induced path $v_i-u_j-v_k-v_j$.
If $v_k$ is non-adjacent to $v_i$ and $v_j$, then $v_k$ belongs to the induced path $v_i-u_j-v_k-u_i-v_j$.
Moreover, the vertex $u_k$ belongs to the induced path $v_i-u_k-v_j$, and the vertices $u$, $u_i$ e $u_j$ belong to the induced path $v_i-u_j-u-u_i-v_j$.

From this, let $S$ be a subset in general position on the monophonic convexity with at least 3 vertices. Therefore, $S$ must induce a clique in $H$ and cannot have two vertices $v_i$ and $u_i$ for the same $i$. Also, $S$ cannot have two vertices $u_i$ and $u_j$ such that $v_i$ and $v_j$ are non-adjacent. Moreover, $S$ cannot have three vertices $u_i$, $u_j$ and $v_k$ for distinct $i,j,k$, since $v_k$ belongs to the induced path $u_i-v_k-u_j$.

With this, we conclude that $S$ consists of (a) a clique of $G$ (formed by a clique of $H$ and a vertex $u_k$), or (b) a clique of $H$ and the vertex $u$, or (c) $S\subseteq\{u_1,\ldots,u_n\}$ is such that $\{v_i:\ u_i\in S\}$ induces a clique of $H$.
Therefore, $H$ has a clique of size $k$ if and only if $G$ has a subset in general position on the monophonic convexity with size $\ell=k+1$.
\end{proof}

\begin{corollary}\label{teo-mono3}
Let $G$ be a graph with diameter two and let $k$ be a positive integer.
The problem of deciding whether $G$ has a subset of size at least $k$ in general position on the monophonic convexity is $\W[1]$-hard when parameterized by the size $k$ of the solution. Moreover, there is no polynomial time algorithm with approximation factor $n^{1-\varepsilon}$ to compute the maximum size $\gp_{\mc}(G)$ of a subset of $G$ in general position on the monophonic convexity, for any $\varepsilon>0$, unless $\P=\NP$.
\end{corollary}

\begin{proof}
These results come directly from the reduction of the previous theorem, since it is also an $\FPT$ reduction which also preserves approximation and from the facts that \textbf{Clique} is \W[1]-hard \cite{Downey2012} and $n^{1-\varepsilon}$-inapproximable for any $\varepsilon>0$, unless $\P=\NP$ \cite{zuckerman06}.
\end{proof}

%-------------------------------------------------------
%-------------------------------------------------------
%-------------------------------------------------------
\section{General position number in the P3 convexity}

In this section, we focus on the general position number in the $P_3$ convexity.
Note that a subset $S\subseteq V(G)$ is in general position in the $P_3$ convexity if and only if $G[S]$ is a graph of maximum degree 1. Thus, it can be seen as a natural generalization of the notion of Independent Set, which induces a subgraph of maximum degree equal to zero. 
As one may expect, this has been studied in the literature under different names. Probably the most referred one is as a \emph{dissociation set}. The \emph{dissociation number} of $G$, denoted by $\diss(G)$, is the maximum cardinality of a dissociation set of $G$, which is then equivalent to $\gp_{\pc}(G)$. Another equivalent definition in the literature is the notion of $1$-dependent set. A \emph{$k$-dependent set} is a subset $S\subseteq V(G)$ such that $\Delta(G[S])\leq k$. The \emph{$k$-dependence number} of a graph $G$ is the cardinality of a maximum $k$-dependent set of $G$. For a survey, see~\cite{Chellali2012}. Another related notion is that of a \emph{3-path cover}. A subset $S\subseteq V(G)$ is a \emph{$k$-path cover} of $G$ if the vertex set of any path on at least $k$ vertices of $G$ is intersected by $S$. Note that $S$ is in general position in the $P_3$ convexity if and only if $V(G)\setminus S$ is a 3-path cover.

\paragraph{Related results.} In the sequel, we present results with respect to the previous notions, but we translate them to the context of the general position number in the $P_3$ convexity.
In~\cite{Yannakakis1981,boliac2004computing}, it is proved that computing $\gp_{\pc}(G)$ is $\NP$-hard even in bipartite graphs, but polynomial-time solvable in bipartite graphs with no induced ``skew star''. In~\cite{papadimitriou1982complexity}, it is proved that the same holds for planar graphs with maximum degree 4. In~\cite{cameron2006independent}, polynomial-time algorithms for chordal graphs, weakly
chordal graphs, asteroidal triple-free graphs and interval-filament graphs are presented. Orlovich et al.~\cite{Orlovichetal2011} proved that $\gp_{\pc}(G)$ is $\NP$-hard even in planar line graphs of a planar bipartite graph with maximum degree 4. They also present polynomial-time algorithms for restricted graph classes. Some other polynomial-time algorithms to particular graph classes can be found in~\cite{LOZIN2003167}.

In~\cite{TSUR20191}, it is presented an algorithm with running time $\mathcal{O}^*(1.713^k)$ to decide whether $G$ has a 3-path cover of cardinality at most $k$, which is the same as to ask whether $\gp_{\pc}(G)\geq n-k$. Thus, note that the dual problem of the general position number in the $P_3$ convexity is $\FPT$, parameterized by $k$.
In~\cite{HOSSEINIAN202295}, a $4/3$-approximation algorithm to the general position number in the $P_3$ convexity is obtained. Such result was further studied by~\cite{BOCK2022160}, which also presented other upper and lower bounds for $\gp_{\pc}(G)$.

In the following, we prove some results regarding the Parameterized Complexity of the $P_3$ general position number $\gc_{\pc}(G)$. First, none of the reductions cited above is a parameterized one. We did not find in the literature the study of the parameterized complexity of determining $\gp_{\pc}(G)\geq k$ parameterized by $k$. We start proving the $\W[1]$-hardness of deciding whether $\gc_{\pc}(G)\geq k$, parameterized by $k$. Let us define the problem we reduce to ours.

In the \textsc{Multicolored Independent Set} problem, the instance is a graph $G$ and a positive integer $k$. Each vertex of $G$ has a color in $\{1,\ldots, k\}$ and the goal is to find an independent set $S$ of $G$ with $k$ vertices, one of each color. Such independent set is called \emph{multicolored $k$-independent set}. This problem is well-known to be $\W[1]$-hard when parameterized by $k$~\cite{cygan2015parameterized}.

\begin{theorem}
\label{thm:w-hard-gppc}
 Deciding whether $\gp_{\pc}(G)\geq k$ parameterized by $k$ is $\W[1]$-hard.
\end{theorem}
\begin{proof}
    Let $(G,k)$ be an instance of \textsc{Multicolored Independent Set}. We build a graph $G'$ from $G$ in linear time such that $\gp_{\pc}(G')\geq 2k$ if and only if $G$ has a multicolored $k$-independent set.

Let $S_i$ be the vertices of $G$ with color $i$ for every $i\in\{1,\ldots, k\}$. To build $G'$ first convert each $S_i$ to a clique (often in the literature this is already a hypothesis in the given instance). Then, add $k$ vertices $u_i$, one for each $i\in\{1,\ldots, k\}$, and make $u_i$ adjacent to all vertices in $S_i$. This finishes the construction of $G'$. Clearly, this construction can be done in linear time. 

    Let us now prove that $\gp_{\pc}(G')\geq 2k$ if and only if $G$ has a multicolored $k$-independent set.

Suppose first that $S$ is a multicolored $k$-independent set of $G$. Let $s_i$ be the vertex in $S\cap S_i$. Define $S'$ as $S\cup \{u_i\mid i\in\{1,\ldots,k\}\}$. Note that, since $S$ is a multicolored $k$-independent set and the only neighbors of $u_i$ are the vertices in $S_i$, $G[S']$ is a graph whose edges are $u_is_i$ for $i\in\{1,\ldots, k\}$, and induce a matching. Thus $S'$ is a subset with $2k$ vertices in general position in the $P_3$ convexity in $G'$ and thus $\gp_{\pc}(G')\geq 2k$.

Let $S'$ be a subset of $V(G')$ in general position in the $P_3$ convexity with at least $2k$ elements. By definition, note that at most 2 vertices of $S'$ may lie in a same clique of $G'$. Since $V(G')$ can be partitioned into $k$ cliques $C_i=S_i\cup \{u_i\}$ for $i\in\{1,\ldots, k\}$, we deduce that $|S'| = 2k$ and that $|S'\cap C_i| = 2$ for each $i \in\{1,\ldots,k\}$.
In case there are two vertices $s_i,s_i' \in S'\cap S_i$ such that $s_i\neq u_i$ and $s_i'\neq u_i$, note that  $(S'\setminus \{s_i'\})\cup \{u_i\}$ is also in general position in the $P_3$ convexity, because $N(u_i) = S_i$. Consequently, we may assume, w.l.o.g., that $u_i\in S'$ for every $i\in\{1,\ldots, k\}$. Therefore, $S = S'\setminus \{u_i\mid i\in\{1,\ldots,k\}\}$ must be a multicolored $k$-independent set of $G$, as any edge linking two of these vertices would induce a path on 4 vertices in $G[S']$, which is not possible as $\Delta(G[S'])\leq 1$.
\end{proof}

After proving Theorem~\ref{thm:w-hard-gppc}, the natural question is answered by the following proposition.

\begin{proposition}
    Deciding whether $\gp_{\pc}(G)\geq k$ parameterized by $k$ is in $\XP$.
\end{proposition}
\begin{proof}
    Note that $\gp_{\pc}\geq k$ if and only if there is a subset $S\subseteq V(G)$ on exactly $k$ vertices in general position in the $P_3$ convexity. Thus, one may verify all the $\binom{n}{k}$ subsets and verify in polynomial time for each one whether it is in general position in the $P_3$ convexity.
\end{proof}

\paragraph{Structural Parameters.} Since we dealt with the parameterized complexity of determining whether $\gp_{\pc}(G)\geq k$ parameterized only by $k$, let us now study other parameters related to the structure of $G$.

A graph $G$ has \emph{neighborhood diversity} at most $d$ if $V(G)$ can be partitioned into $d$ sets of twins~\cite{cygan2015parameterized}. In particular, note that each part must correspond to a clique or an independent set. Moreover, for every pair of parts, either there are all edges linking vertices from one part to the vertices in the other, or none. Denote by $\nd(G)$ the minimum cardinality of such partition of $G$.

\begin{theorem}
\label{thm:polykernelnd}
    Deciding whether $\gp_{\pc}(G)\geq k$ has a kernel of size $\mathcal{O}(\nd(G)\cdot k)$ when parameterized by the neighborhood diversity $\nd(G)$ of $G$ plus $k$.
\end{theorem}
\begin{proof}
    Assume $k\geq 3$ as otherwise the problem can be solved in polynomial time.
    One can first obtain in polynomial time such partition of $V(G)$ into $\nd(G)$ sets of twins by just checking, for each pair of vertices, whether they are twins. Then, if one part corresponds to an independent set of cardinality at least $k$, return ``YES'', as such independent set is a set in general position in the $P_3$ convexity. While there is a part corresponding to a clique that has at least three vertices, then one may remove all but two vertices in this part as at most two of them can be part of any optimal solution and they are twins. Thus, we obtain an equivalent instance with at most $\nd(G)(k-1)$ vertices.
\end{proof}

Recall that a \emph{vertex cover} $S$ in a graph $G$ is a set of vertices such that each edge in $E(G)$ has at least one extremity in $S$. The cardinality of a minimum vertex cover in $G$ is denoted by $\vc(G)$. Determining whether $\vc(G)\leq k$ is one of Karp's 21 $\NP$-complete problems~\cite{Karp1972}, but admits a 2-approximation algorithm~\cite{cygan2015parameterized}. By definition, it is also well-known that $S$ is a vertex cover of $G$ if and only if $V(G)\setminus S$ is an independent set of $G$.

Note that if $G$ satisfies $\vc(G)\leq k$, then $\nd(G)\leq 2^k+k$. Indeed, if $S$ is a minimum vertex cover of $G$, one can partition the independent set $V(G)\setminus S$ into at most $2^k$ subsets of (false) twins (the ones with the same neighborhood in $S$) and complete the partition with the singletons of $S$. Thus, if one considers the problem of deciding whether $\gc_{\pc}\geq k$ parameterized by $\vc(G) +k$, then Theorem~\ref{thm:polykernelnd} implies an exponential kernel. We can easily improve such kernel. 

\begin{theorem}
\label{thm:linearkernelvc}
The problem of deciding whether $\gp_{\pc}(G)\geq k$ has a kernel of size $\mathcal{O}(\vc(G)+k)$ when parameterized by $\vc(G)+k$, where $\vc(G)$ is the size of a minimum vertex cover of $G$.
\end{theorem}

\begin{proof}
One can first use the 2-approximation algorithm for Vertex Cover~\cite{cygan2015parameterized} to obtain in polynomial time a vertex cover $S\subseteq V(G)$ such that $|S|\leq 2\cdot \vc(G)$. Then, if $|V(G)\setminus S|\geq k$, then return ``YES''. Otherwise, $|V(G)|\leq 2\cdot \vc(G)+k-1$.
\end{proof}

Finally, we obtain parameterized results from two algorithmic meta-theorems in graphs with bounded local-treewidth and graphs with bounded cliquewidth. Given a graph $G$, let \textsc{GP-Dec}$_{\pc}(k)$ be the problem of deciding whether $G$ has a $P_3$ general position set of size $k$.

The local-treewidth \cite{Eppstein00} of a graph $G$ is the function $ltw_G:\NN\to\NN$ which associates with any $r\in\NN$ the maximum treewidth of an $r$-neighborhood in $G$. That is, $ltw_G(r)=\max_{v\in V(G)}\{tw(G[N_r(v)])\}$, where $N_r(v)$ is the set of vertices at distance at most $r$ from $v$. We say that a graph class $\mathcal{C}$ has bounded local-treewidth if there is a function $f_\mathcal{C}:\NN\to\NN$ such that, for all $G\in\mathcal{C}$ and $r\in\NN$, $ltw_G(r)\leq f_\mathcal{C}(r)$.
It is known that graphs with bounded genus or bounded max degree have bounded local-treewidth \cite{Eppstein00}. In particular, a graph with max degree $\Delta$ has $ltw_G(r)\leq \Delta^r$ and a planar graph has $ltw_G(r)\leq 3r-1$ \cite{bodlaender98}.

In the following, we express the \textsc{GP-Dec}$_{\pc}(k)$ decision problem in First Order logic. We use lower case variables $x,y,z,\ldots$ (resp. upper case variables $X,Y,Z,\ldots$) to denote vertices (resp. subsets of vertices) of a graph. The \emph{atomic formulas} are $x=y$, $x\in X$ and $E(x,y)$ which denotes the adjacency relation in a given graph. The Boolean connectives are $\wedge,\vee,\neg,\to$ and $\leftrightarrow$ and the quantifiers are $\exists$ and $\forall$. Let MSOL be the monadic second order logic (with quantification over subsets of vertices) and let FO be the first order logic (with quantification over vertices).

\begin{theorem}\label{teo-fo}
Given a graph $G$ having bounded local-treewidth, the \textsc{GP-Dec}$_{\pc}(k)$ decision problem is FPT when parameterized by the size $k$ of the solution. More precisely, it can be solved in $\mathcal{O}(f(k) \cdot n^2)$ time.
\end{theorem}

\begin{proof}
Consider the following first order formula \textsc{gp-Set}$_{\pc}(X)$, where $X\subseteq V(G)$, which is true if and only if $X$ is a $P_3$ general position set of $G$:
\[
\mbox{\textsc{gp-Set}}_{\pc}(X)\ :=\ \forall u,v,w\in X :\  \Big(\ \big(E(u,v) \wedge (u\neq w)\big) \rightarrow \neg E(v,w)\Big).
\]
Therefore the decision problem \textsc{GP-Dec}$_{\pc}(k)$ is FO expressible from the following formula of size $\mathcal{O}(k^2)$:
\[
\mbox{\textsc{GP-Dec}}_{\pc}(k) := \exists v_1,\ldots,v_k\in V(G):\ \mbox{\textsc{gp-Set}}_{\pc}(\{v_1,\ldots,v_k\})\wedge\bigwedge_{1\leq i<j\leq k} v_i\ne v_j
\]
Then, from the Frick-Grohe Theorem (see \cite{courcelle2}), \textsc{GP-Dec}$_{\pc}(k)$ is FPT with parameter $k$ in time $\mathcal{O}(n^2)$ for graphs with bounded local treewidth.
\end{proof}

Finally, we prove the following theorem on bounded cliquewidth graphs, such as cographs (cliquewidth 2), distance hereditary graphs (cliquewidth 3), $(q,q-4)$-graphs and bounded treewidth graphs.

\begin{theorem}
Deciding whether $\gp_{\pc}(G)\geq k$ is $\FPT$ parameterized by cliquewidth of $G$.
Moreover, the problem of finding a maximum $P_3$ general position set is polynomial time solvable in bounded cliquewidth graphs.
\end{theorem}

\begin{proof}
In graphs, an optimization problem is LinEMSOL if it wants to maximize (or minimize) some linear function over the sizes of $\ell$  subsets (of vertices), which satisfy an MSOL formula, for fixed $\ell$.
In \cite{courcelle1}, it was proved that LinEMSOL optimization problems are polynomial time solvable in graphs with bounded cliquewidth. The running time is linear if a cliquewidth expression is given. Moreover, a cliquewidth expression can be obtained in cubic time in bounded cliquewidth graphs.

Since the $P_3$ general position optimization problem wants to obtain the maximum subset $X\subseteq V(G)$ satisfying the formula \textsc{GP-Dec}$_{\pc}(k)$ described previously, which is a FO and an MSOL formula, then the maximization problem is polynomial time solvable in graphs with bounded cliquewidth.

Moreover, given a positive integer $k$, it is possible to decide if $\gp_{\pc}(G)\geq k$ in cubic time for bounded cliquewidth graphs, just running the cubic time algorithm to obtain a cliquewidth expression and running the linear time algorithm for the optimization problem and checking if $k$ is at most the maximum value. This implies that deciding whether $\gp_{\pc}(G)\geq k$ is $\FPT$ parameterized by the cliquewidth of $G$.
\end{proof}

As a consequence, we have the following corollary.

\begin{corollary}
The problem of obtaining a maximum general position set in the $P_3$ convexity is linear time solvable in distance-hereditary graphs.
\end{corollary}

\begin{proof}
It is known that a cliquewidth expression of any distance hereditary graph can be obtained in
linear time \cite{golumbic00}. From \cite{courcelle2}, LinEMSOL optimization problems are linear time solvable in graphs with bounded cliquewidth, if a cliquewidth expression is given, and we are done.
\end{proof}

\section*{Acknowledgments}

The authors were partially supported by CNPq [305404/2020-2], [311070/2022-1] and [404479/2023-5], CAPES [88881.197438/2018-01] and [88881.712024/2022-01],  FUNCAP [186-155.01.00/2021] and FAPERJ [211.753/2021].

\bibliographystyle{plain}
%\bibliography{refs-v3}

\end{document}